\documentclass[12pt, a4paper,leqno]{amsart}
\usepackage{amsmath,amsthm,amscd,amssymb,amsfonts,amsbsy}
\usepackage{latexsym}
\usepackage{txfonts, dsfont}
\usepackage{exscale}

\usepackage[numbers]{natbib}

\usepackage[colorlinks,citecolor=red,pagebackref,hypertexnames=false]{hyperref}
\usepackage[foot]{amsaddr}

\calclayout
\allowdisplaybreaks

\numberwithin{equation}{section}

\usepackage{pgf}
\usepackage{color}

\theoremstyle{plain}
\newtheorem{theorem}[equation]{Theorem}
\newtheorem{lemma}[equation]{Lemma}
\newtheorem{corollary}[equation]{Corollary}

\newtheorem{hypothesis}[equation]{Hypothesis}

\theoremstyle{definition}

\newtheorem{definition}[equation]{Definition}

\theoremstyle{remark}

\newtheorem{example}[equation]{Example}

\newcommand{\R}{\mathbb{R}}

\newcommand{\Conn}{{W_c}}

\newcommand{\RR}{{\mathbb{R}}}
\newcommand{\NN}{{\mathbb{N}}}
\newcommand{\ZZ}{{\mathbb{Z}}}

\newcommand{\eps}{\varepsilon}

\newcommand{\bp}{\noindent {\it Proof}.\,\,}
\newcommand{\ep}{\hfill$\Box$ \vskip 0.08in}

\newcommand{\vp}{\varphi}

\renewcommand{\emptyset}{\varnothing}

\begin{document}

\title{The effective potential of an $M$-matrix}

\author{Marcel Filoche}
\address{Marcel Filoche, Laboratoire de Physique de la Matière Condensée, Ecole Polytechnique, CNRS, IP Paris, 91128, Palaiseau, France}
\author{Svitlana Mayboroda}
\address{Svitlana Mayboroda, School of Mathematics, University of Minnesota, Minneapolis, Minnesota, 55455, USA}
\author{Terence Tao}
\address{Terence Tao, Department of Mathematics, UCLA, Los Angeles CA 90095-1555, USA}

\date{}

\begin{abstract}
In the presence of a confining potential $V$, the eigenfunctions of a continuous Schr\"odinger operator $-\Delta +V$ decay exponentially with the rate governed by the part of $V$ which is above the corresponding eigenvalue; this can be quantified by a method of Agmon. Analogous localization properties can also be established for the eigenvectors of a discrete Schr\"odinger matrix. This note shows, perhaps surprisingly, that one can replace a discrete Schr\"odinger matrix by \emph{any} real symmetric $Z$-matrix and still obtain eigenvector localization estimates.  In the case of a real symmetric non-singular $M$-matrix $A$ (which is a situation that arises in several contexts, including random matrix theory and statistical physics), the \emph{landscape function} $u = A^{-1} 1$ plays the role of an effective potential of localization. Starting from this potential, one can create an Agmon-type distance function governing the exponential decay of the eigenfunctions away from the ``wells" of the potential, a typical eigenfunction being localized to a single such well. 
 \end{abstract}

\maketitle

%\ms\noindent{\bf Key words/Mots cl\'es.}

%\ms\noindent
% AMS classification: 49Q20 35B65.

%\tableofcontents

\section{Introduction: history and motivation}

The fundamental premises of quantum physics guarantee that a potential $V$ induces exponential decay of the eigenfunctions of the Schr\"odinger operator $-\Delta +V$ (on either a continuous domain $\RR^d$ or a discrete lattice $\ZZ^d$) as long as $V$ is larger than the eigenvalue $E$ outside of some compact region. This heuristic principle has been established with mathematical rigor by S. Agmon~\cite{AgmonBook} and has served as a foundation to many beautiful results in semiclassical analysis and other fields (see, e.g., Refs~\citenum{H,HS,S1,S2} for a glimpse of some of them). Roughly speaking, the modern interpretation of this principle is that the eigenfunctions decay exponentially away from the ``wells" $\{x: V(x)\leq E\}$.

In 2012, two of the authors of the present paper introduced the concept of the \emph{localization landscape}. They observed in Ref.~\citenum{FM-PNAS} that the solution $u$ to the equation $(-\Delta +V)u=1$ appears to have an almost magical power to ``correctly" predict the regions of localization for disordered potentials $V$ and to describe a precise picture of their exponential decay. For instance, if $V$ takes the values $0$ and $1$ randomly on a two-dimensional lattice $\ZZ^2$ (a classical setting of the Anderson--Bernoulli localization) the eigenfunctions at the bottom of the spectrum are exponentially localized, that is, exponentially decaying away from some small region, but this would not be detected by the Agmon theory because the region $\{V\leq E\}$ could be completely percolating and there is no ``room" for the Agmon-type decay, especially if the probability of $V=0$ is larger than the probability of $V=1$. And indeed, the phenomenon of Anderson localization is governed by completely different principles, relying on the interferential rather than confining impact of $V$. On the other hand, looking at the landscape in this example, we observe that the region $\{\frac{1}{u}\leq E\}$ exhibits isolated wells and that the eigenmodes decay exponentially away from these wells. It turns out that indeed, the reciprocal of the landscape, $\frac{1}{u}$, plays the role of an {\it effective potential}, and in Ref.~\citenum{ADFJM-math} Arnold, David, Jerison, and the first two authors have proved that the eigenfunctions of $-\Delta +V$ decay exponentially in the regions where $\{\frac{1}{u}>E\}$ with the rate controlled by the so-called \emph{Agmon distance} associated to the landscape, a geodesic distance in the manifold determined by $(\frac{1}{u}-E)_+$. The numerical experiments in Ref.~\citenum{ADFJM-comp} and physical considerations in Ref.~\citenum{ADFJM-phy} show an astonishing precision of the emerging estimates, although mathematically speaking in order to use these results for factual disordered potentials one has to face, yet again, a highly non-trivial question of resonances -- see the discussion in Ref.~\citenum{ADFJM-math}. At this point we have only successfully treated Anderson potentials via the localization landscape in the context of a slightly different question about the integrated density of states~\cite{LL}.

However, the scope of the landscape theory is not restricted to the setting of disordered potentials. In fact, all results connecting the eigenfunctions to the landscape are purely deterministic, and one of the key benefits of this approach is the absence of \emph{a priori} assumptions on the potential $V$, which already in Ref.~\citenum{ADFJM-math} allowed us to rigorously treat any operator $-\operatorname{div} A \nabla +V $ with an elliptic matrix of bounded measurable coefficients $A$ and any non-negative bounded potential $V$, a level of the generality not accessible within the classical Agmon theory. These ideas and results have been extended to quantum graphs~\cite{HM}, to the tight-binding model~\cite{WZh}, and perhaps most notably, to many-body localization in Ref.~\citenum{Galitski}. 

This paper shows that the applicability of the landscape theory in fact extends well beyond the scope of the Schr\"odinger operator, or, for that matter, even the scope of PDEs, at least in the bottom of the spectrum where the region $\{\frac{1}{u} \leq E \}$ exhibits isolated potential wells. Indeed, let us now consider a general real symmetric positive definite $N \times N$ matrix $A = (a_{ij})_{i,j \in [N]}$, which one can view as a self-adjoint operator on the Hilbert space $\ell^2([N])$ on the domain $[N] \coloneqq \{1,\dots,N\}$. In certain situations one expects $A$ to exhibit ``localization'' in the following two related aspects, which we describe informally as follows:
\begin{itemize}
\item[(i)]  (Eigenvector localization) Each eigenvector
\footnote{One can also study the closely related phenomenon of localization of Green's functions $(A-z)^{-1}$.  This latter type of localization is also related to the spectrum of associated infinite-dimensional operators consisting of pure point spectrum, thanks to such tools as the Simons--Wolff criterion~\cite{sw}.}
$\phi = (\phi_k)_{k \in [N]}$ of $A$ is localized to some index $i$ of $[N]$, so that $|\phi_k|$ decays when $|k-i|$ exceeds some localization length $L \ll N$.
\item[(ii)]  (Poisson statistics) The local statistics of eigenvalues $\lambda_1,\dots,\lambda_N$ of $A$ asymptotically converge to a Poisson point process in a suitably rescaled limit as $N \to \infty$.
\end{itemize}
Empirically, the phenomena (i) and (ii) are observed to occur in the same matrix ensembles $A$; intuitively, the eigenvector localization property (i) implies that $A$ ``morally behaves like'' a block-diagonal matrix, with the different blocks of $A$ supplying ``independent'' sets of eigenvalues, thus leading to the Poisson statistics in (ii).  However, the two properties (i), (ii) are not formally equivalent; for instance, conjugating $A$ by a generic unitary matrix will most likely destroy property (i) without affecting property (ii).   

\begin{example}[Gaussian band matrices]\label{gbm}  Consider the random band matrix Gaussian models $A$, in which the entries $a_{ij}$ are independent Gaussians for $1 \leq i \leq j \leq N$ and $|i-j| \leq W$, but vanish for $|i-j| > W$, for some $1 \leq W \leq N$.  We refer the reader to Ref.~\citenum{Bourgade} (\S 2.2) for a recent survey of this model.  If the matrix is normalized to have eigenvalues $E$ concentrated in the interval $[-2,2]$ (and expected to obey the Wigner semicircular law $\displaystyle \frac{1}{2\pi} (4-E^2)_+^{1/2}$ for the asymptotic density of states), it is conjectured (see e.g., Ref.~\citenum{FM}) that the localization length $L$ should be given by the formula
$$ L \sim \min( W^2 (4-E^2), N );$$
in particular, in the bulk of the spectrum, it is conjectured that localization (in both senses (i), (ii)) should hold when $W \ll N^{1/2}$ (with localization length $L \sim W^2$) and fail when $W \gg N^{1/2}$, while near the edge of the spectrum (in which $4-E^2 = O(L^{-2/3})$)  localization is expected to hold when $W \ll N^{5/6}$ and fail when $W \gg N^{5/6}$.  Towards this conjecture, it is known~\cite{BYY} in the bulk $4-E^2 \sim 1$ that (i), (ii) both fail when $W \gg N^{3/4+\eps}$ for any fixed $\eps>0$, while localization in sense (i) was established for $W \ll N^{1/7}$ in Ref.~\citenum{peled} (see also Ref.~\citenum{Sh}).  In the edge $4-E^2 = O(L^{-2/3})$, both directions of the conjecture have been verified in sense (ii) in Ref.~\citenum{Sodin}, but the conjecture in sense (i) remains open. Finally, we remark that in the regime $W=O(1)$ the classical theory of Anderson localization~\cite{Anderson1958} can be used to establish both (i) and (ii).
\end{example}

We now focus on the question of establishing eigenvector localization (i).  Can one deduce any uniform bound on the eigenvectors of a general matrix $A$ depicting, in particular, a structure of the exponential decay similarly to the aforementioned considerations for a matrix of the Schr\"odinger operator $-\Delta +V$? An immediate objection is that there is no ``potential" that could play the role of $V$. Even aside from the fact that the proof of the Agmon decay relies on the presence of both kinetic and potential energy, as well as on many PDE arguments, it is not clear whether there is a meaningful function, analogous to $V$, which governs the behavior of eigenvectors of a general matrix. The main result of this paper is that, surprisingly, the landscape theory still works, at least in the class of real symmetric $Z$-matrices (matrices with non-positive entries off the diagonal).  Furthermore, when $A$ is a real symmetric non-singular $M$-matrix (a positive semi-definite $Z$-matrix), the reciprocal $\frac{1}{u}$ of the solution to $Au=1$ gives rise to a distance function $\rho$ on the index set $[N]$ which predicts the exponential decay of the eigenvectors.

\section{Main results}

We introduce an Agmon-type distance $\rho$ on the index set $[N] \coloneqq \{1,\dots,N\}$ associated to an $N \times N$ matrix $A$, a $N \times 1$ ``landscape'' vector $u$, and an additional spectral parameter $E \in \R$:

\begin{definition}[Distance]\label{distance-def}  Let $A = (a_{ij})_{i,j \in [N]}$ be a real symmetric $N \times N$ matrix, let $u= (u_i)_{i \in [N]}$ be a vector with all entries non-zero, and let $E$ be a real number.  We define the \emph{effective potential} $\overline{V} = (\overline{v}_i)_{i \in [N]}$ by the formula
\begin{equation}\label{vv-def}
\overline{v}_i\coloneqq \frac{(Au)_i}{u_i},
\end{equation}
the \emph{shifted effective potential} by the formula
\begin{equation}\label{vi-def}
v_i \coloneqq \left( \overline{v}_i - E \right)_+
\end{equation}
(where $x_+ \coloneqq \max(x,0)$), the \emph{potential well set} by the formula
$$ K_E \coloneqq \{ i \in [N]: v_i = 0 \} = \left\{ i \in [N]: \frac{(Au)_i}{u_i} \leq E \right\},$$
and the \emph{distance function} $\rho = \rho_{A,u,E} \colon [N] \times [N] \to [0,+\infty]$ by the formula
$$\rho(i,j) \coloneqq \inf_{L \geq 0} \inf_{i_0,\dots,i_L \in [N]: i_0=i, i_L = j} \left( \sum_{\ell=0}^{L} \ln \left(1+ \sqrt{\frac{\sqrt{v_{i_\ell} v_{i_{\ell+1}}}}{|a_{i_\ell i_{\ell+1}}|}} \right) \right)$$
where we restrict the infimum to those paths $i_0,\dots,i_L$ for which $a_{i_\ell i_{\ell+1}} \neq 0$ for $\ell=0,\dots,L-1$.  
To put it another way, $\rho$ is the largest pseudo-metric such that
\begin{equation}\label{rhoj}
 \rho(i,j) \leq \ln \left(1+ \sqrt{\frac{\sqrt{v_i v_j}}{|a_{ij}|}} \right) 
\end{equation}
whenever $a_{ij} \neq 0$.

For any set $M\subset [N]$ we denote by $\rho(i,M) \coloneqq  \inf_{j\in M}\rho(i, j)$ the distance from a given index $i$ to $M$ using the distance $\rho$ (with the convention that $\rho(i,M)=\infty$ if $M$ is empty).  Similarly, for any set $K \subset [N]$, we define $\rho(K,M) \coloneqq \inf_{i \in K} \rho(i,M)$ for the separation between $K$ and $M$.
\end{definition}

It is easy to see that $\rho$ is a pseudo-metric in the sense that it is symmetric and obeys the triangle inequality with $\rho(i,i)=0$, although without further hypotheses\footnote{For instance, if we assume that $A$ is \emph{irreducible} in the sense that it cannot be expressed (after permuting indices) as a block-diagonal matrix, then $\rho(i,j)$ will always be finite.} on $A,u,E$ it is possible that $\rho(i,j)$ could be zero or infinite for some $i \neq j$. One can view $\rho$ as a weighted graph metric on the graph with vertices $[N]$ and edges given by those pairs $(i,j)$ with $a_{ij} \neq 0$, and with weights given by the right-hand side of \eqref{rhoj}. We discuss the comparison between $\rho$ and the Euclidean metric in the beginning of Section~\ref{resonant}. 

We recall that a \emph{$Z$-matrix} is any $N \times N$ matrix $A$ such that $a_{ij} \leq 0$ when $i\neq j$, and a \emph{$M$-matrix} is a $Z$-matrix with all eigenvalues having non-negative real part. Our typical set-up is the case when $A$ is a real symmetric non-singular $M$-matrix, i.e., a positive definite matrix with non-positive off-diagonal entries, and in that case we will choose $u$ as the landscape function, i.e., the solution to $Au=1$, with $1$ denoting a vector with all values equal to~$1$.  We say that a matrix $A$ has \emph{connectivity} at most $\Conn$ if every row and column has at most $\Conn$ non-zero non-diagonal entries. If $A$ is a real symmetric non-singular $M$-matrix, all the principal minors are positive (see e.g., Ref.~\citenum{plemmons}), and hence by Cram\'er's rule all the coefficients of the landscape $u$ will be non-negative.  In this case, a simple form of our main results is as follows.

\begin{theorem}[Exponential localization using landscape function]\label{c2.3}
Let $A$ be a symmetric $N\times N$ $M$-matrix with connectivity at most $\Conn$ for some $\Conn \geq 2$.
Let $u \coloneqq A^{-1} 1$ be the landscape function.  Assume that $\vp$ is an $\ell^2$-normalized eigenvector of $A$ corresponding to the eigenvalue $E$. Let $\rho = \rho_{A,u,E}$, $K = K_E$ be defined by Definition \ref{distance-def}.
Then
$$\sum_k \vp_k^2 \,e^{\frac{2\rho(k,K)}{\sqrt{\Conn}} }\, \left( \frac{1}{u_k} -E \right)_+ \leq \Conn\max_{1 \leq i,j \leq N} |a_{ij}|.$$
\end{theorem}

Informally, the above inequality ensures that an eigenvector $\vp$ experiences exponential decay away from the wells of the effective potential $\overline{V} = (\frac{1}{u_k})_{k \in [N]}$ cut off by the energy level $E$. This is what typically happens for the Schr\"odinger operator  $-\Delta + V$ (according to some version of the Agmon theory); see for instance \cite[Corollary 4.5]{ADFJM-math}.  However, the existence of such an effective potential for an \emph{arbitrary} $M$-matrix is perhaps surprising. 

In fact our results apply to the larger class of real symmetric $Z$-matrices $A$ and more general vectors $u$, and can handle ``local'' eigenvectors as well as ``global'' ones. We first introduce some more notation.

\begin{definition}[Local eigenvectors]\label{local} Let $M\subset [N]$.  We use $I_M$ to denote the $N \times N$ diagonal matrix with $(I_M)_{ii}$ equal to $1$ when $i \in M$ and $0$ otherwise.  If $\vp \in \ell^2([N])$, we write $\vp|_M \coloneqq I_M \vp$ for the restriction of $\vp$ to $M$ (extending by zero outside of $M$), and similarly if $A$ is an $N \times N$ matrix we write $A|_M \coloneqq I_M A I_M$ for the restriction of $A$ to $M \times M$ (again extending by zero).  We say that a vector $\vp \in \ell^2([N])$ is a \emph{local eigenvector} of $A$ on the domain $M$ with eigenvalue $E$ if $\vp = \vp|_M$ is an eigenvector of $A|_M$ with eigenvalue $E$, thus $I_M \vp=\vp$ and $I_M A I_M \vp = E \vp$.
\end{definition}

To avoid confusion we shall sometimes refer to the original notion of an eigenvector as a \emph{global eigenvector}; this is the special case of a local eigenvector in which $M = [N]$.

We can now state a more general form of Theorem \ref{c2.3}.

\begin{theorem}[Exponential localization]\label{t1.1.1}
Let $A$ be a symmetric $N\times N$ $Z$-matrix with connectivity at most $\Conn$ for some $\Conn \geq 2$, and let $u$ be some $n\times 1$ vector of non-negative coefficients. Let $\overline{E}>0$ be an energy threshold, and let $\rho = \rho_{A,u,\overline E}$, $v_i$, and $K_{\overline{E}}$ be defined by Definition \ref{distance-def}.
Then for any subset $D$ of $[N]$ and any local eigenvector $\vp$ of $A$ of eigenvalue $E \leq \overline{E}$ on $D^c=[N]\setminus D$, one has
\begin{multline} \label{eq1.1.2}
(\overline E -E) \sum_{k \notin K_{\overline{E}}} |\vp_k|^2 e^{2\alpha \rho(k,K_{\overline{E}}\setminus D)} \\
+ \left(1-\frac{\alpha^2 \Conn}{2}\right)\sum_{k \notin K_{\overline{E}}} |\vp_k|^2 e^{2\alpha \rho(k,K_{\overline{E}}\setminus D)} v_k \\
\leq   \frac {\Conn}{2}  \|\vp\|^2 \max_{i\in K_{\overline{E}}\setminus D, \, j\not\in K_{\overline{E}}\setminus D} |a_{ij}|, 
\end{multline}
for any $0<\alpha\leq \sqrt{2/\Conn}$.  (Here and in the sequel we use $\| \cdot \|$ to denote the $\ell^2([N])$ norm.) 

In particular, if $\alpha = \sqrt{1/\Conn}$, $\overline E =E$, $D = \emptyset$, and $\vp$ is an $\ell^2$-normalized (global) eigenvector of $A$ on the entire domain $[N]$ with the eigenvalue $E$, \eqref{eq1.1.2} implies that
\begin{equation}\label{gen}
\sum_{k \in [N]} \vp_k^2 \,e^{\frac{2\rho(k,K_E)}{\sqrt{\Conn}} }\, v_k \leq \Conn \max_{i,j \in [N]} |a_{ij}|.
\end{equation}
\end{theorem}

There are two terms on the left-hand side of \eqref{eq1.1.2}, corresponding to two different lines in the display, and they serve different purposes. The bound for the term in the second line (which in particular yields \eqref{gen}) asserts roughly speaking that the eigenvector $\vp_k$ experiences exponential decay in the regime where $k$ is far from $K_{\overline{E}}$ in the sense that $\rho(k,K_{\overline{E}}) \gg \sqrt{\Conn}$.  Note that Theorem \ref{c2.3} is the special case of \eqref{gen} when $A$ is a $M$-matrix and $u=A^{-1} 1$.

By taking advantage of the term in the first line of \eqref{eq1.1.2}, we can proceed further and demonstrate an approximate \emph{diagonalization}, or \emph{decoupling}, of $A$ on the collection of disjoint subregions defined by the landscape function $u$, by following the arguments from Ref.~\citenum{ADFJM-math}. The details are too technical to be put in the introduction, and we refer the reader to Section~\ref{resonant}. In short, the idea is that viewing $[N]$ as a graph induced by $A$ (with the vertices connected whenever $a_{ij}\neq 0$), we can define a Voronoi-type splitting of this graph into subgraphs, $\Omega_\ell$, each containing an individual connected component of $K_{\overline{E}}$ (or sometimes merging a few components if convenient). Then $A$ can be essentially decoupled into smaller matrices $A|_{\Omega_\ell}$ with the strength of coupling exponentially small in the $\rho_{A,u}$ distance between individual ``wells''. Related to this, the spectrum of $A$ will be exponentially close to the combined spectrum of $A|_{\Omega_\ell}$'s.

Note how the geometry of the metric $\rho$ is sensitive to the spatial distribution of the matrix $A$, and in particular to the connectivity properties of the graph induced by the locations of the nonzero locations of $A$. For instance, conjugating $A$ by a generic orthogonal matrix will almost certainly destroy the localization of the eigenvectors $\vp$, but will also heavily scramble the metric $\rho$ (and most likely also destroy the property of being an $M$-matrix or $Z$-matrix). On the other hand, conjugating $A$ by a permutation matrix will simply amount to a relabeling of the (pseudo-)metric space $([N],\rho)$, and not affect the conclusions of Corollary \ref{c2.3} and the decoupling results in Theorem~\ref{tDiag} and Corollary~\ref{C4.3} in any essential way.

We will show some results of the numerical simulations in the next section, and then pass to the proofs, but let us say a few more words about the particular cases which would perhaps be of most interest. 

{\bf Random band matrices.}  Here the connectivity is $\Conn = 2W$.  Strictly speaking, the random Gaussian band matrix models $A$ considered in Example \ref{gbm} do not fall under the scope of Corollary \ref{c2.3}, because the matrices will not be expected to have non-positive entries away from the diagonal, nor will they be expected to be positive definite.  However, one can modify the model to achieve these properties (at least with high probability), by replacing the Gaussian distributions by distributions supported on the negative real axis, and then shifting by a suitable positive multiple of the identity to ensure positive definiteness with high probability.  These changes will likely alter the semicircle law for the bulk distribution of eigenvalues, but in the spirit of the universality phenomenon, one may still hope to see localization of eigenvectors, say in the bulk of the spectrum, as long as the width $W$ of the band matrix is small enough (in particular when $W \ll N^{1/2}$).  In this case Corollary~\ref{c2.3} entails exponential decay of the eigenvectors governed by the landscape $\frac{1}{u}$ and Theorem~\ref{tDiag} and Corollary~\ref{C4.3} yield the corresponding diagonalization of $A$. Of course, the key question is the behavior of the landscape. If the set $K_{\overline{E}}$ of wells is localized to a short interval, then this corollary will establish localization in the spirit of (i) above; however, if $K_{\overline{E}}$ is instead the union of several widely separated intervals then an eigenvector could in principle experience a resonance in which non-trivial portions of its $\ell^2$ energy were distributed amongst two or more of these intervals. Whether or not this happens is governed to some extent by Theorem~\ref{tDiag} and Corollary~\ref{C4.3}. These results indicate that the resonances have to be exponentially strong in the distance between the wells, and our numerical experiments suggest that such strong resonances are in fact quite rare.   

{\bf Tight-binding Schr\"odinger operators.} When $A$ is a matrix of the tight-binding Schr\"odinger operator (a standard discrete Laplacian plus a potential) in a cube in $\ZZ^d$, the connectivity parameter $\Conn$ is now the number of nearest neighbors, $2d$, and the size of the matrix is the sidelength of the cube to the power $d$. If the potential is non-negative, $A$ is an $M$-matrix with the entries $a_{ij}$ equal to $-1$ whenever $i\neq j$ corresponds to the nearest neighbors in the graph structure induced by $\ZZ^d$, and $a_{ii}=2d+V_i$. This particular case has been considered in Ref.~\citenum{WZh} and our results clearly cover it. However, the tight-binding Schr\"odinger is only one of many examples, even when concentrating on applications in physics. We can treat any operator in the form $-\mathrm{div} A \nabla+V$ on any graph structure, provided that the signs of the coefficients yield an $M$-matrix. We can also address long range hopping for a very wide class of Hamiltonians.

{\bf Many-body system and statistical physics.} Much more generally, in statistical physics, the probability distribution over all possible microstates (or the density matrix in the quantum setting) of a given system evolves through elementary jumps between microstates. This evolution is a Markov process whose transition matrix is a $Z$-matrix which is symmetric up to a multiplication by a diagonal matrix. For a micro-reversible evolution, the matrix $A$ is symmetric and is akin to a weighted Laplacian on the high-dimensional indirect graph whose vertices are the microstates and whose edges are the possible transitions. 

One essential result of statistical physics is that, under condition of irreducibility of the transition matrix, the system eventually reaches thermodynamical equilibrium. Our approach might open the way to unravel the structure of the eigenvectors of the Markov flow, and thus to understand how localization of these eigenvectors can induce a many-body system to remain ``frozen" for mesoscopic times out of equilibrium. This effect is referred to as \emph{many-body localization}. A first successful implementation of the landscape theory in this context has been recently achieved by V.~Galitski and collaborators~\cite{Galitski} for a many-body system of spins with nearest-neighbor interaction. In this work, the authors cleverly use the ideas of Ref.~\citenum{BAA} to transfer the problem to the Fock space and to deduce an Agmon-type decay governed by the corresponding effective potential. Once in the Fock space, their results are also a particular case of Theorem~\ref{t1.1.1} and Theorem~\ref{tDiag}. From that point, however, the authors of Ref.~\citenum{Galitski} go much farther to discuss, based on physical considerations, deep implications of such an exponential decay on many-body localization, but in the present paper we restrict ourselves to mathematics and will not enter those dangerous waters.

Finally, we would like to mention that the idea of trying the localization landscape and similar concepts in the generality of random matrices has appeared before, e.g., in Refs~\citenum{St} and \citenum{B}. However, the authors relied on a different principle, extending the inequality $|\vp|\leq E u$ from Ref.~\citenum{FM-PNAS} to these more general contexts, which by itself, of course, does not prove exponential decay. Ref.~\citenum{St} actually deals with a different proxy for the landscape and different inequalities, but we (and the authors) believe that these are related to the landscape and that, again, they do not prove exponential decay estimates. However, we would like to mention that the importance of $M$-matrices was already suggested in Ref.~\citenum{B}, and it was inspiring and reassuring to arrive at the same setting from such different points of view.

\section{Numerical simulations}\label{sNum}

We ran numerical simulations to compute the localization landscape $u$, the effective potential $\frac{1}{u}$, and the eigenvectors for several realizations of random symmetric $M$-matrices. The diagonal coefficients are random variables which follow a centered normal law of variance 1. The off-diagonal coefficients belonging to the first $\Conn/2$ diagonals of the upper triangle of the matrix are minus the absolute values of random variables following the same law. The remaining off-diagonal coefficients of the upper triangle are taken to be zero, and the lower triangle is completed by symmetry. This creates $A_0$, a $Z$-matrix of bandwidth $\Conn+1$ (and connectivity $\Conn$). To ensure positivity, we add a multiple of the identity 
\begin{equation}
A: = A_0 + a~I \qquad \textrm{where} \qquad  a = \varepsilon - \lambda_0 ~,
\end{equation}
$\lambda_0$ being the smallest eigenvalue of $A_0$ and $\varepsilon = 0.1$. The smallest eigenvalue of the resulting matrix $A$ is thus $\varepsilon$. The matrices $A$ and $A_0$ clearly have the same eigenvectors and their spectra differ only by a constant shift.

Below are the results of several simulations. We take $N=10^3$. Figures~\ref{fig:Wc_2}-\ref{fig:Wc_32} correspond to random symmetric $M$-matrices constructed as above of connectivity $\Conn=2$, $6$, $20$, and $32$. Each figure consists of two frames:

The top frame displays the localization landscape $u$ superimposed with the first $5$~eigenvectors plotted in $\log_{10}$ scale. The exponential decay of the eigenvectors can clearly be observed on this frame for $\Conn=2$, $6$, and $20$. One can see that, as expected, it starts disappearing around $\Conn=32$ ($\Conn$ being in this case roughly equal to $\sqrt{N}$). It is important to observe that in all cases the eigenvectors decay exponentially except for the wells of $\frac{1}{u}$ (equivalently, the peaks of $u$) where they stay flat. This is exactly the prediction of  Theorem~\ref{t1.1.1}.

The bottom frame displays the effective potential $\frac{1}{u}$ superimposed with the first $5$~eigenvectors plotted in linear scale. The horizontal lines indicate the energies of the corresponding eigenvectors. One can clearly see the localization of the eigenvectors inside the wells of the effective potential.

\begin{figure}
\centering
\includegraphics[width=0.85\textwidth]{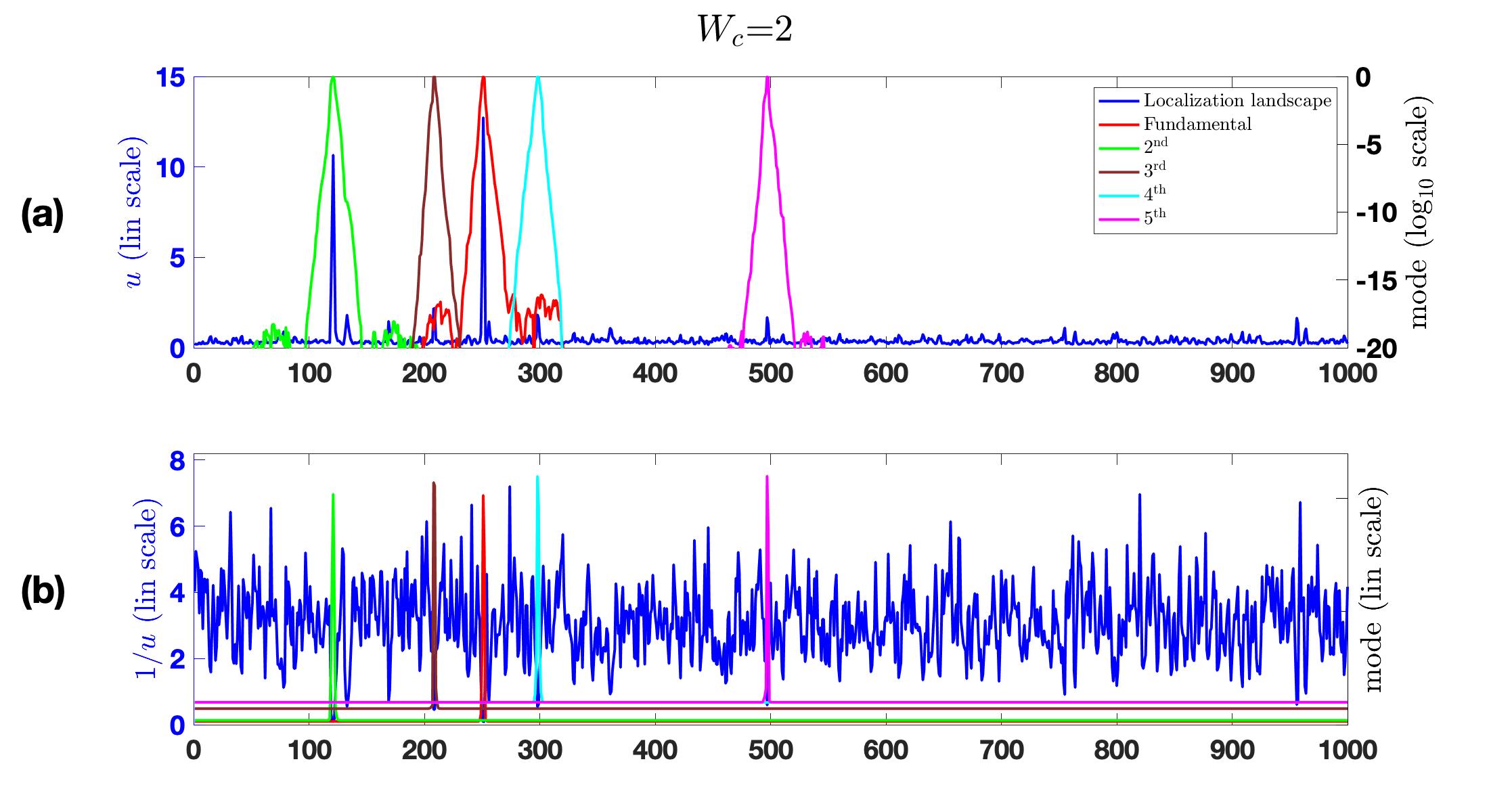}
\caption{(a) Localization landscape (blue line) and the $5$ first eigenvectors (in $\log_{10}$ scale) for a random $3$-band symmetric $M$-matrix. (b) Effective potential ($\frac{1}{u}$) and the first eigenvectors (in linear scale). The baseline (the 0 of the vertical axis) is chosen differently for each eigenvector so that it coincides with the eigenvalue of the same eigenvector of the left axis. This convention will be used in all Figures 1 to 4.}
\label{fig:Wc_2}
\end{figure}

\begin{figure}
\centering
\includegraphics[width=0.85\textwidth]{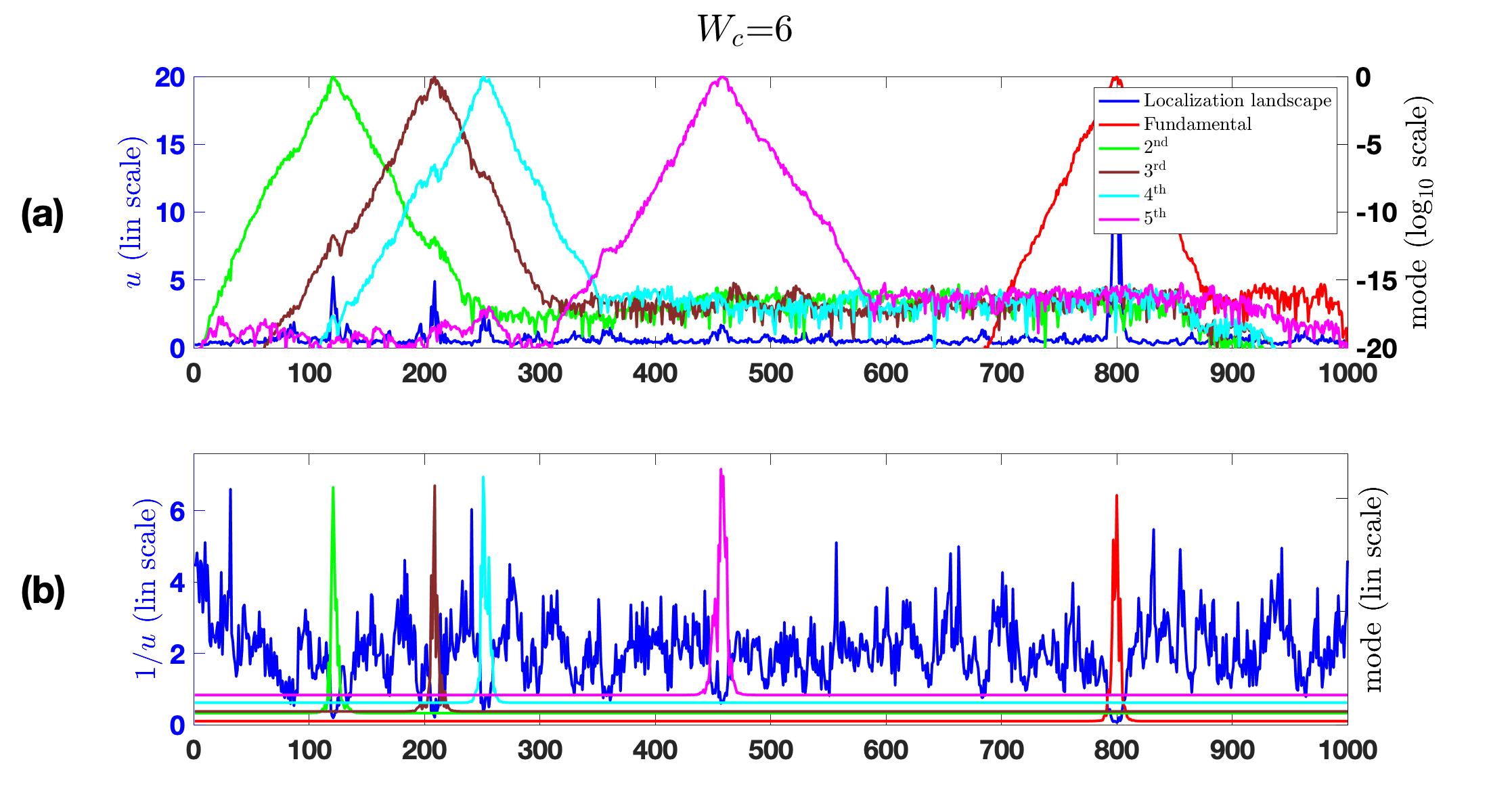}
\caption{(a) Localization landscape (blue line) and the $5$ first eigenvectors (in $\log_{10}$ scale) for a random $7$-band symmetric $M$-matrix. (b) Effective potential ($\frac{1}{u}$) and the first eigenvectors (in linear scale)}
\label{fig:Wc_6}
\end{figure}

\begin{figure}
\centering
\includegraphics[width=0.85\textwidth]{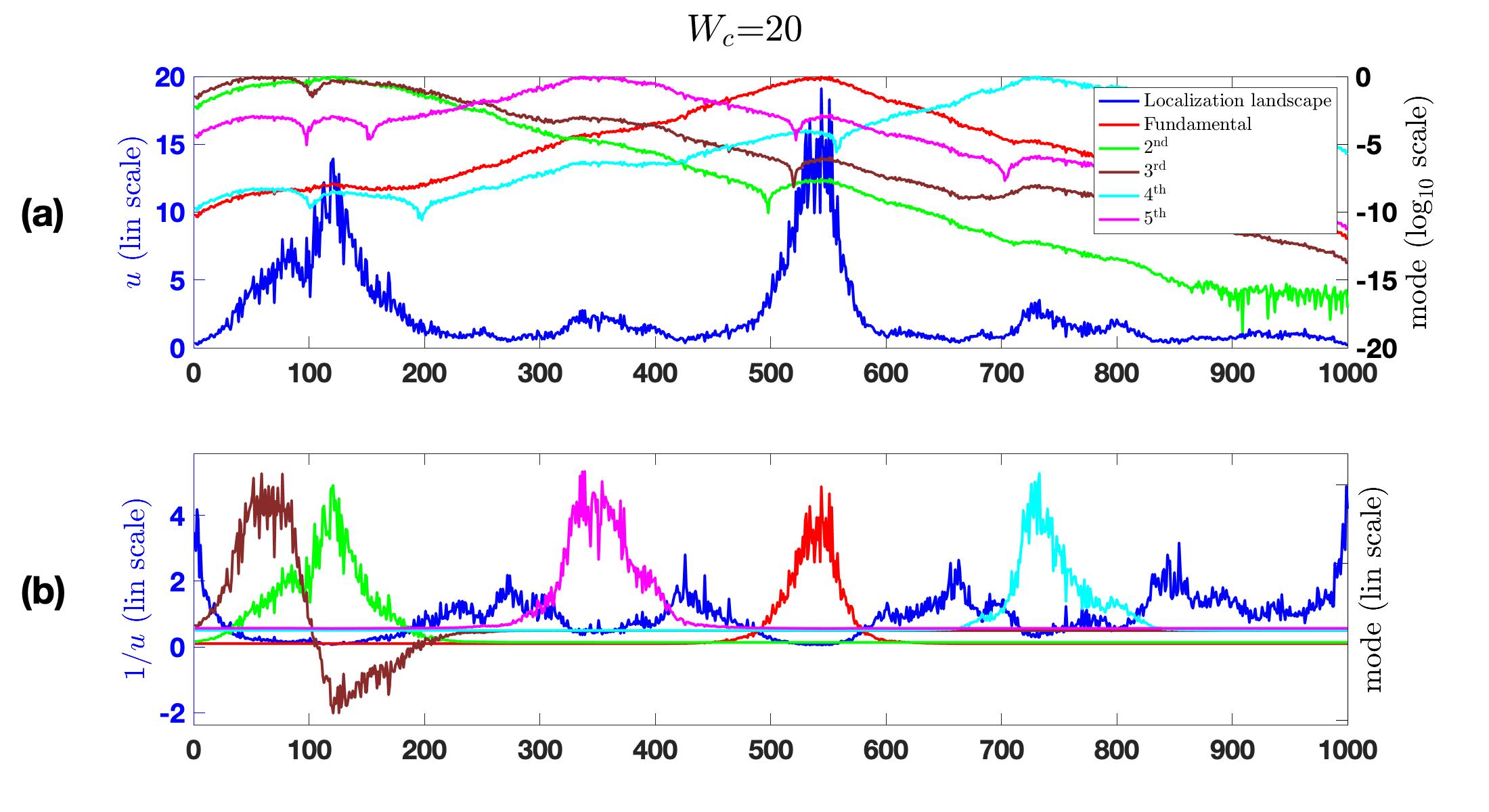}
\caption{(a) Localization landscape (blue line) and the $5$ first eigenvectors (in $\log_{10}$ scale) for a random $21$-band symmetric $M$-matrix. (b) Effective potential ($\frac{1}{u}$) and the first eigenvectors (in linear scale).}
\label{fig:Wc_20}
\end{figure}

\begin{figure}
\centering
\includegraphics[width=0.85\textwidth]{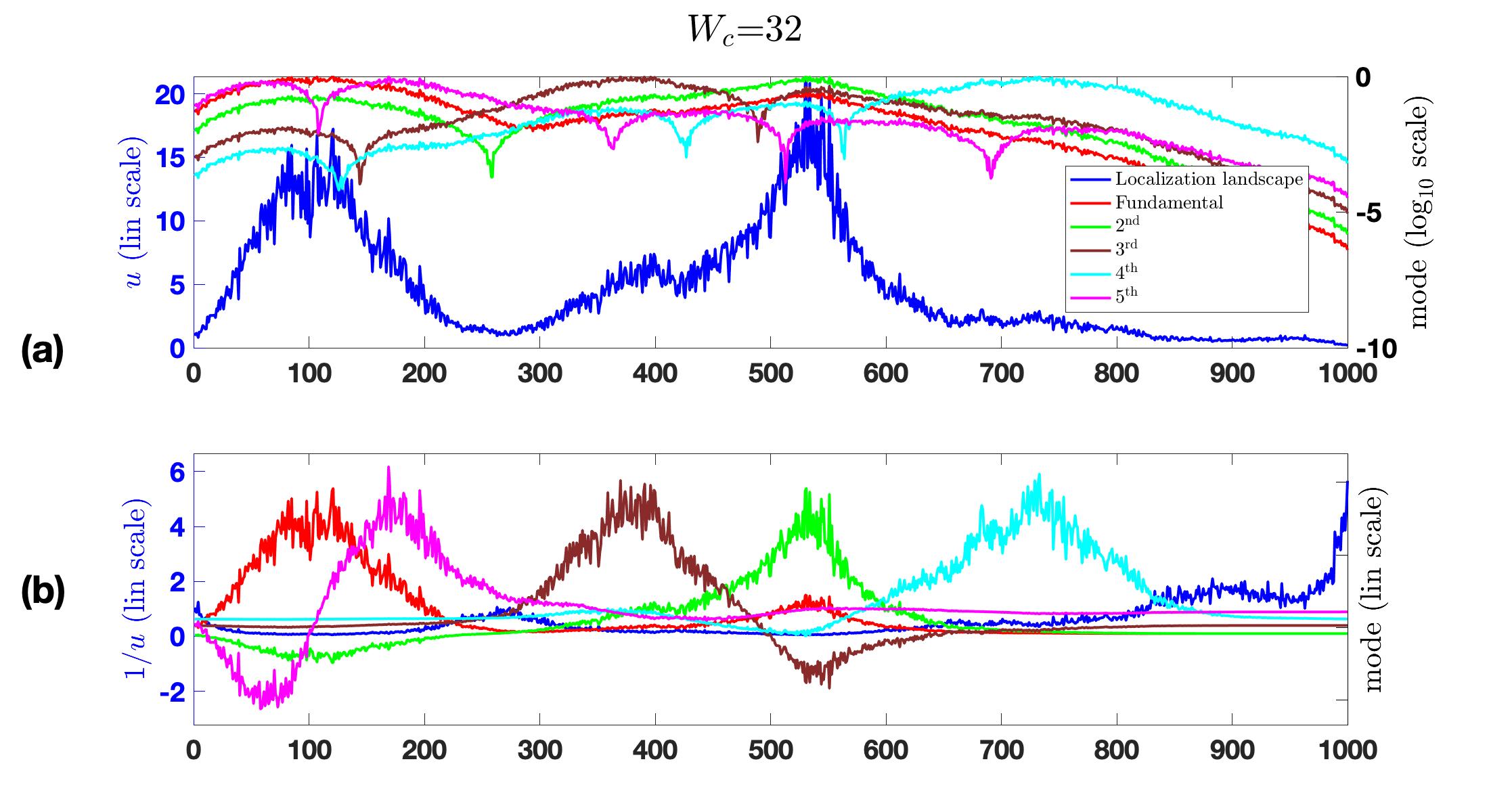}
\caption{(a) Localization landscape (blue line) and the $5$ first eigenvectors (in $\log_{10}$ scale) for a random $33$-band symmetric $M$-matrix. (b) Effective potential ($\frac{1}{u}$) and the first eigenvectors (in linear scale).}
\label{fig:Wc_32}
\end{figure}

Figure~\ref{fig:Agmon} provides numerical evidence for finer effects encoded in Theorems~\ref{t1.1.1} and \ref{tDiag}. The two Theorems combined prove exponential decay away from the wells of the effective potential governed the Agmon distance associated to $1/u$, at least in the absence of resonances. In Figure~\ref{fig:Agmon} we display, for several values of connectivity $\Conn$ and several eigenvectors, the values $-\ln|\psi_i|$ against the distance $\rho_{A,u, E}(i, i_{\max})$, taking as the origin the point $i_{\max}$ where $|\psi|$ is maximal, and using the corresponding eigenvalue as the threshold $E$. The linear correspondence down to $e^{-40}$ is quite remarkable and shows that  $e^{-c\rho_{A,u, E}(i, i_{\max})}$ is not only an upper bound, but actually an approximation of the eigenfunction, and that the resonances are indeed unlikely. On the other hand, the constant $c$ does not appear to be equal to $1/\sqrt{\Conn}$ which means that in this respect our analysis is probably not optimal, at least in the class of random matrices. Indeed, we believe that the application of the deterministic Schur test in the proof does not yield the best possible constant for random coefficients, but since we emphasize the universal deterministic results, this step cannot be further improved.

\begin{figure}
\centering
\includegraphics[width=0.95\textwidth]{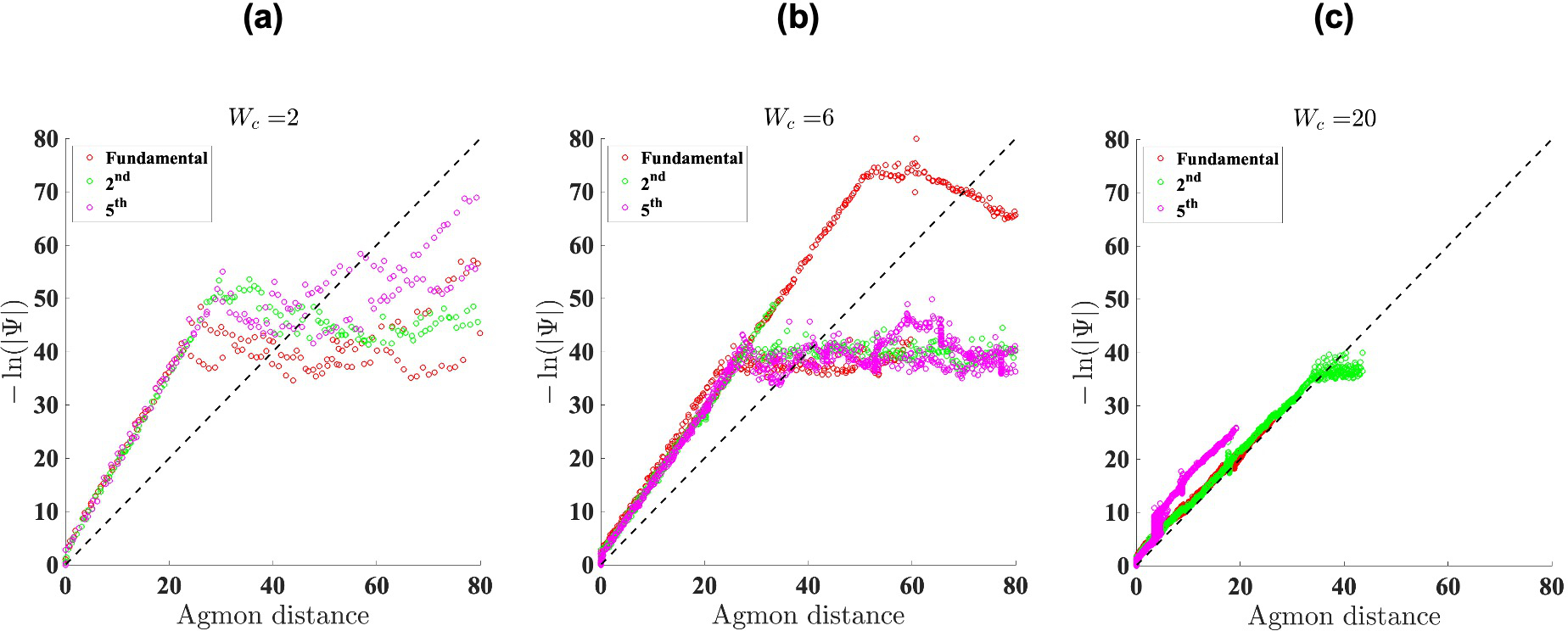}
\caption{Scatter plots of the logarithm of the absolute value of several eigenvectors against the corresponding Agmon distance, for 3~different values of the connectivity~$\Conn=2$, 6, and 20 (frames (a), (b), and (c)). For each eigenvector (eigenvectors \#1, 2, and 5 in each frame), we display $-\ln |\psi_i|$ at any given point~$i$  vs. the Agmon distance between the point~$i$ and the location where $|\psi|$ is maximal. The plots exhibit a strong linear relationship between these two quantities, down to values of $|\psi_i|$ around $e^{-40}$ (of the order of $10^{-18}$), which is a signature of the exponential decay. The slope seems to depend only on $W_c$.}
\label{fig:Agmon}
\end{figure}

Finally, Figure~\ref{fig:wells} shows that Hypothesis~\ref{separation} is actually fulfilled in some realistic situations. The top frame displays the example already presented in Figure~\ref{fig:Wc_6}. Superimposed to the eigenvectors, the set $K_E$ introduced in Definition~\ref{distance-def} is also drawn (grey rectangles) for $E=0.7$ (horizontal dashed red line). The middle frame displays the plot of the Agmon distance to $K_E$. Thresholding this plot at $S=2$ (green horizontal line) allows us to draw the $S$-neighborhood of $K_E$ (the orange rectangles). The bottom frame shows a possible partition of the entire domain into five subdomains ($\Omega_1, \cdots, \Omega_5$), each subdomain containing at least one well of the effective potential $1/u$. The distances $\rho(\partial^- \Omega_\ell, K_\ell)$ defined in Hypothesis~\ref{separation} are here respectively 48.1584, 2.8093, 4.2169, 3.6784, 9.3756. They all are larger than $S$, thus satisfying  Hypothesis~\ref{separation} .

\begin{figure}
\centering
\includegraphics[width=0.9\textwidth]{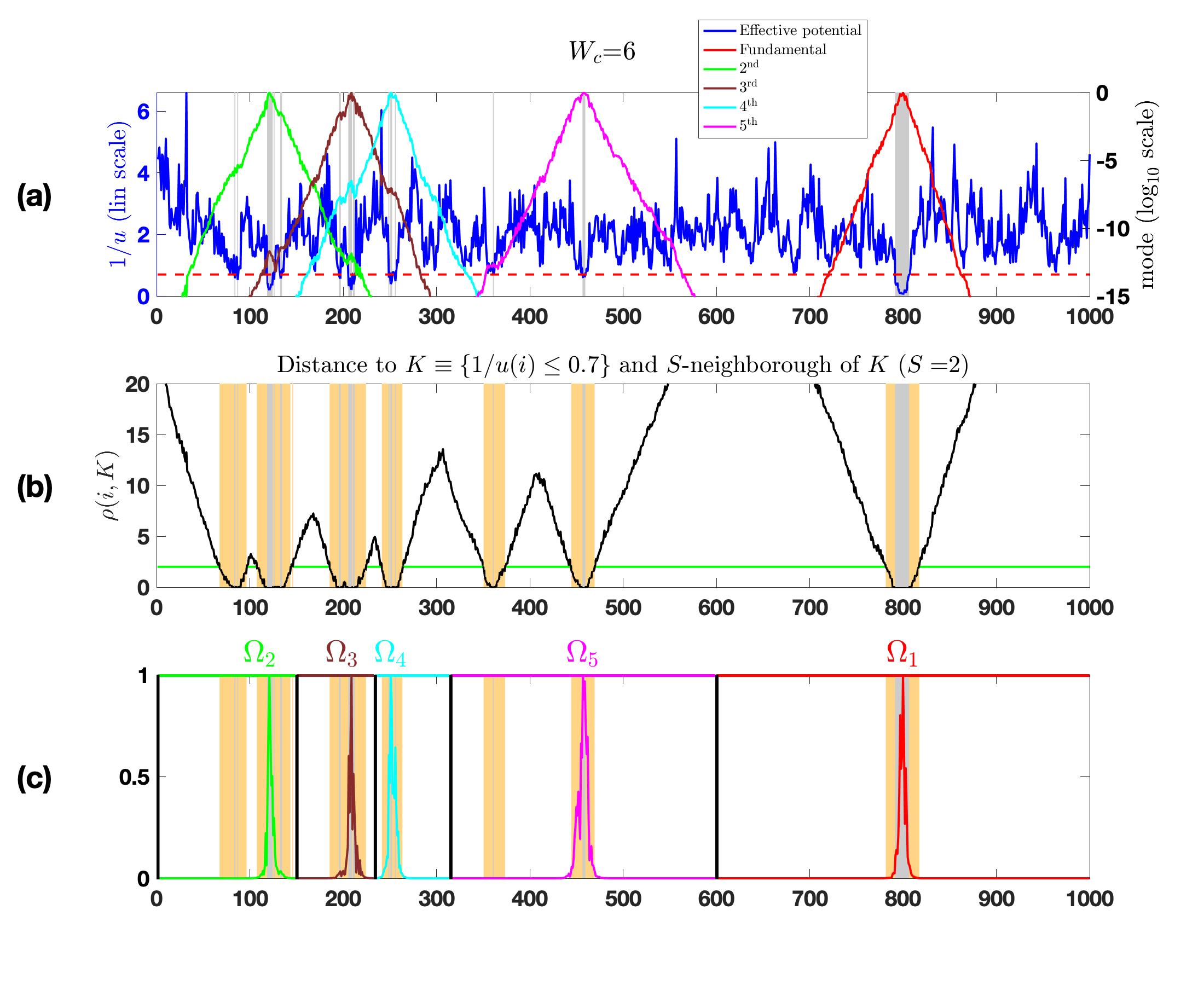}
\caption{(a) Eigenvectors in log scale, superimposed with the set $K_E$ defined in \ref{distance-def} (grey rectangles) for the value $E=0.7$ (indicated by the red dashed line). (b) Plot of the Agmon distance of each point to the set~$K_E$. The orange rectangles correspond to the $S$-neighborhood of $K_E$ for $S=2$ (indicated by the green horizontal line). (c) Partition of the domain in five subdomains. All distances $\rho(\partial^- \Omega_\ell, K_\ell)$ defined in~\ref{separation} are larger than $S$. This partition thus fulfills Hypothesis~\ref{separation}.}
\label{fig:wells}
\end{figure}

To be more precise, let us turn to the exact statements.

\clearpage

\section{The proof of the main results}

In this section we prove Theorem \ref{t1.1.1}.  We will use a double commutator method.  Let $[A,B] \coloneqq AB-BA$ denote the usual commutator of $N \times N$ matrices, and $\langle, \rangle$ the usual inner product on $\ell^2([N])$.  We observe the general identity
\begin{equation}\label{direct}
 \langle [[A,D],D] u, u \rangle = \sum_{i,j \in [N]: i \neq j} a_{ij} u_i u_j (d_{ii} - d_{jj})^2.
\end{equation}
whenever $A = (a_{ij})_{i,j \in [N]}$ is a matrix, $D = \mathrm{diag}(d_{11},\dots,d_{nn})$ is a diagonal matrix, and $u = (u_i)_{i \in [N]}$ is a vector.  In particular we have 
\begin{equation}\label{neg}
\langle [[A,D],D] u, u \rangle \leq 0
\end{equation}
whenever $A$ is a $Z$-matrix and the entries of $u$ have constant sign.  It will be this negative definiteness property that is key to our arguments.  One can compare \eqref{direct}, \eqref{neg} to the Schr\"odinger operator identity
$$
\langle [[-\Delta + V, g],g] u, u \rangle = - 2 \int_{\R^d} |\nabla g|^2 |u|^2 \leq 0$$
for any (sufficiently well-behaved) functions $V, g, u: \R^d \to \R$.

To exploit \eqref{direct} we will use the following identity.

\begin{lemma}[Double commutator identity]\label{dc} Let $A, \Psi, G$ be $N \times N$ real symmetric matrices such that $\Psi G = G \Psi$, and suppose that $u$ is an $N \times 1$ vector.  Then
$$ \langle G [\Psi,A] u, G \Psi u \rangle = \frac{1}{2} \langle [[A,G\Psi],G\Psi] u, u \rangle - \frac{1}{2} \langle [[A,G],G] \Psi u, \Psi u \rangle.$$
\end{lemma}

\begin{proof} By the symmetric nature of $G$ we have
$$ \langle [[A,G],G] \Psi u, \Psi u \rangle = 2\langle GA\Psi u, G\Psi u\rangle - 2 \langle A G \Psi u, G \Psi u \rangle$$
and similarly from the symmetric nature of $G \Psi$ we have
$$ \langle [[A,G\Psi],G\Psi] u, u \rangle = 2 \langle G \Psi A u, G \Psi u \rangle - 2 \langle A G \Psi u, G \Psi u \rangle.$$
The claim follows.
\end{proof}

We can now conclude

\begin{corollary}\label{dc-cor} Let $A = (a_{ij})_{i,j \in [N]}$ be a $N \times N$ real symmetric $Z$-matrix. Assume that $D$ is some subset of $[N]$ and that $\vp$ is a local eigenvector of $A$ corresponding to the eigenvalue $E$ on $D^c=[N]\setminus D$. Let $u = (u_i)_{i \in [N]}$ be a vector with all positive entries, and let $G = \mathrm{diag}(G_{11},\dots,G_{NN})$ be a real diagonal matrix.  Then
\begin{equation}\label{dc-eq}
 \sum_{k \in [N]} \varphi_k^2 G_{kk}^2 \left( \frac{(Au)_k}{u_k} - E \right) \leq -\frac{1}{2} \sum_{i,j \in [N]: i \neq j} a_{ij} \varphi_i \varphi_j (G_{ii} - G_{jj})^2.
\end{equation}
\end{corollary}

\begin{proof}  Writing $[\Psi,A] = \Psi(A-EI) - (A-EI)\Psi$, we apply Lemma \ref{dc} with $\Psi \coloneqq \mathrm{diag}( \varphi_1/u_1,\dots,\varphi_N/u_n )$, to get
\begin{multline*}
 \langle G \Psi (A-E I) u, G \Psi u \rangle- \langle G (A-E I) \Psi  u, G \Psi u \rangle 
 \\= \frac{1}{2} \langle [[A,G\Psi],G\Psi] u, u \rangle - \frac{1}{2} \langle [[A,G],G] \Psi u, \Psi u \rangle.
\end{multline*}
By \eqref{neg} the first term on the right-hand side above is non-positive and hence, the entire expression is less than or equal to 
$$
 - \frac{1}{2} \langle [[A,G],G] \Psi u, \Psi u \rangle
\leq -\frac{1}{2} \sum_{i,j \in [N]: i \neq j} a_{ij} \varphi_i \varphi_j (G_{ii} - G_{jj})^2.
$$
The latter inequality follows from \eqref{direct} and the fact that by definition $\Psi u = \varphi$.

Since $G$ is diagonal, and $\vp=\Psi u$ is a local eigenvector on $D^c$, the second term on the left-hand side is equal to zero. Indeed, $(\Psi u)_k=(G\Psi u)_k=0$ for $k\in D$, and hence 
$$\langle G (A-E I) \Psi  u, G \Psi u \rangle=\langle G (A-E I)|_{D^c} \Psi  u, G \Psi u \rangle=0.$$
Writing
$$ \Psi (A-E I) u = \left( \left( \frac{(Au)_k}{u_k} - E \right) \varphi_k \right)_{k \in [N]}$$
the claim follows.
\end{proof}

The strategy is then to apply this corollary with a sufficiently slowly varying function $G$, so that one can hope to mostly control the right-hand side of \eqref{dc-eq} by the left-hand side.

\noindent {\it Proof of Theorem~\ref{t1.1.1}.} We abbreviate $K_{\overline{E}}$ as $K$ for simplicity. We can of course assume that $\varphi$ is not identically zero.  If $K \setminus D$ was empty we could apply Corollary \ref{dc-cor} with $G_{kk}=1$ to obtain a contradiction, so we may assume without loss of generality that $K \setminus D$ is non-empty.  We apply Corollary \ref{dc-cor} with
$$ G_{ii} \coloneqq 1_{i \not \in K\setminus D} \,e^{\alpha \rho_{A,u, \overline E}(i,K\setminus D)},$$
where the indicator $1_{i \not \in K\setminus D}$ is equal to zero for $i \in K\setminus D$ and equal to $1$ otherwise.  By construction, $G_{kk}$ vanishes for  $k\in K\setminus D$, and $\vp$ vanishes on $D$, so that $G_{kk}\vp_k$ vanishes on $K$. Thus by \eqref{vi-def}
\begin{equation}\label{klo}
\begin{split}
\left(\overline E-E \right) &\sum_{k\notin K} \varphi_k^2 e^{2\alpha \rho(i,K)}+ \sum_{k\notin K} \varphi_k^2 e^{2\alpha \rho(i,K)} v_k \\
 &= \sum_{k\notin K} \varphi_k^2 G_{kk}^2 \left(\overline E-E \right) + \sum_{k\notin K} \varphi_k^2 G_{kk}^2 \left( \frac{(Au)_k}{u_k} - \overline E \right)_+ 
\\
 &= \sum_{k\notin K} \varphi_k^2 G_{kk}^2 \left(\overline E-E \right) + \sum_{k\notin K} \varphi_k^2 G_{kk}^2 \left( \frac{(Au)_k}{u_k} - \overline E \right)
\\
&=  \sum_{k \in [N]} \varphi_k^2 G_{kk}^2 \left( \frac{(Au)_k}{u_k} - E \right) \\
&\leq -\frac{1}{2} \sum_{i,j \in [N]: a_{ij} \neq 0; i \neq j} a_{ij} \varphi_i \varphi_j (G_{ii} - G_{jj})^2.
\end{split}
\end{equation}
Now we need to estimate the quantity $G_{ii} - G_{jj}$ whenever $a_{ij} \neq 0$.  We first observe from the triangle inequality and \eqref{rhoj} that 
\begin{align*}
|e^{\alpha \rho(i,K\setminus D)} - e^{\alpha \rho(j,K\setminus D)}| &\leq e^{\alpha \rho(i,K\setminus D)} \left(e^{\alpha \rho(i,j)} - 1\right) \\
&\leq e^{\alpha \rho(i,K\setminus D)} \left( \left( 1 + \sqrt{\frac{\sqrt{v_i v_j}}{|a_{ij}|}} \right)^\alpha - 1 \right) \\
&\leq e^{\alpha \rho(i,K\setminus D)} \alpha \sqrt{\frac{\sqrt{v_i v_j}}{|a_{ij}|}} 
\end{align*}
and similarly with $i$ and $j$ reversed; in particular
$$
|e^{\alpha \rho(i,K\setminus D)} - e^{\alpha \rho(j,K\setminus D)}|^2 \leq \alpha^2 e^{\alpha \rho(i,K\setminus D)} e^{\alpha \rho(j,K\setminus D)} \frac{\sqrt{v_i v_j}}{|a_{ij}|}.
$$
Thus we have
\begin{equation}\label{g-bound}
(G_{ii} - G_{jj})^2 \leq \alpha^2 e^{\alpha \rho(i,K\setminus D)} e^{\alpha \rho(j,K\setminus D)} \frac{\sqrt{v_i v_j}}{|a_{ij}|}
\end{equation}
when $i,j \not \in K\setminus D$.

Next, suppose that $i \not \in K\setminus D$, $j \in K\setminus D$.  Then $G_{jj}=0$, and from \eqref{rhoj} we have
\begin{align*}
\rho(i,K\setminus D) &\leq \rho(i,j) \\
&\leq \ln \left(1+ \sqrt{\frac{\sqrt{v_i v_j}}{|a_{ij}|}} \right) \\
&= 0
\end{align*}
since $v_j=0$.  We conclude that $(G_{ii}-G_{jj})^2=1$ in this case.  Similarly if $i \in K\setminus D$ and $j \not \in K\setminus D$.  Finally, if $i,j \in K\setminus D$ then $G_{ii}=G_{jj}=0$, so that $(G_{ii}-G_{jj})^2=0$ in this case.  Applying all of these estimates, we can bound the right-hand side of \eqref{klo} by
\begin{equation}\label{new-rhs}
\begin{split}
& \frac{\alpha^2}{2} \sum_{i,j \not \in K\setminus D: i \neq j; \,\,a_{ij} \neq 0} |\varphi_i| |\varphi_j| e^{\alpha \rho(i,K\setminus D)} e^{\alpha \rho(j,K\setminus D)} \sqrt{v_i v_j} \\
&\quad + \frac{1}{2} \sum_{i \in K\setminus D; j \not \in K\setminus D \hbox{ or } i \not \in K\setminus D, j \in K\setminus D; \,\,a_{ij} \neq 0} |a_{ij}| |\varphi_i| |\varphi_j|.
\end{split}
\end{equation}
Since $A$ has at most $\Conn$ non-zero non-diagonal entries in each row and column, we see from Schur's test (or the Young inequality $ab \leq \frac{1}{2} a^2 + \frac{1}{2} b^2$) that
$$ 
\sum_{i,j \not \in K\setminus D: i \neq j; \,\,a_{ij} \neq 0} |\varphi_i| |\varphi_j| e^{\alpha \rho(i,K\setminus D)} e^{\alpha \rho(j,K\setminus D)} \frac{\sqrt{v_i v_j}}{|a_{ij}|}
\leq \Conn \sum_{i \not \in K\setminus D} |\varphi_i|^2 e^{2\alpha \rho(i,K\setminus D)} v_i$$
and
$$
\sum_{i \in K\setminus D; j \not \in K\setminus D \hbox{ or } i \not \in K\setminus D, j \in K\setminus D; a_{ij} \neq 0} |a_{ij}| |\varphi_i| |\varphi_j|
\leq \Conn (\sup_{i \in K\setminus D; j \not \in K\setminus D} |a_{ij}|) \sum_{i \in [N]} |\varphi_i|^2.$$
 Combining all of the above considerations, we arrive at the conclusion of the theorem.
 \ep
 
\section{Diagonalization}
\label{resonant}

Let the notation and hypotheses be as in Theorem \ref{t1.1.1}.  We abbreviate $\rho = \rho_{A,u,\overline{E}}$ and $K = K_{\overline{E}}$.  To illustrate the decoupling phenomenon we place the following hypothesis on the potential well set $K$:

\begin{hypothesis}[Separation hypothesis]\label{separation}  There exists a parameter $S>0$, a partition $K=\bigcup_\ell K_\ell$ of $K$ into disjoint ``wells'' $K_\ell$, and ``neighborhoods'' $\Omega_\ell \supset K_\ell$ of each well $K_\ell$ obeying the following axioms:
\begin{itemize}
\item[(i)]  The neighborhoods $\Omega_\ell$ are all disjoint.
\item[(ii)]  The neighborhoods $\Omega_\ell$ contain the $S$-neighborhood of $K_\ell$, thus $\rho( \Omega_\ell^c, K_\ell ) \geq S$.
\item[(iii)]  For any $\ell$, we have $\rho(\partial^- \Omega_\ell, K_\ell) \geq S$, where the inner boundary $\partial^- \Omega_\ell$ is defined as the set of all $k \in \Omega_\ell$ such that $a_{kj} \neq 0$ for some $j \notin \Omega_\ell$.
\end{itemize}
\end{hypothesis}

We remark that axioms (i), (ii), (iii) imply that the full boundary $\partial \Omega_\ell$, defined as the union of the inner boundary $\partial^- \Omega_\ell$ and the outer boundary $\partial^+ \Omega_\ell$ consisting of those $j \notin \Omega_\ell$ such that $a_{kj} \neq 0$ for some $k \in \Omega_\ell$, stays at a distance at $S$ from $K$, since every element of an outer boundary $\partial^+ \Omega_\ell$ either lies in the inner boundary of another $\Omega_{\ell'}$, or else lies outside of all of the $\Omega_{\ell'}$. We also remark that axiom (iii) is a strengthening of axiom (ii), since if there was an element $k$ in $\Omega_\ell^c$ at distance less than $S$ from $K_\ell$ then by taking a geodesic path from $K_\ell$ to $k$ one would eventually encounter a counterexample to (iii), but we choose to retain explicit mention of axiom (ii) to facilitate the discussion below.

Informally, to obey Hypothesis \ref{separation}, one should first partition $K$ into ``connected components'' $K_\ell$, concatenating two such components together if their separation $\rho$ is too small, so that the separation $\overline{S} \coloneqq \inf_{\ell \neq \ell'} \rho(K_\ell,K_{\ell'})$ is large, and then perform a Voronoi-type partition in which $\Omega_\ell$ consists of those $k \in [N]$ which lie closer to $K_\ell$ in the $\rho$ metric than any other $K_{\ell'}$.  The axioms (i), (ii) would then be satisfied for any $S < \overline{S}/2$ thanks to the triangle inequality, and when $\overline{S}$ is large one would expect axiom (iii) to also be obeyed if we reduce $S$ slightly.  It seems plausible that one could weaken the axiom (iii) and still obtain decoupling results comparable to those presented here, but in this paper we retain this (relatively strong) axiom in order to illustrate the main ideas. 

We have already demonstrated in Section~\ref{sNum} non-vacuousness of Hypothesis \ref{separation}, at least in typical numerical examples. Let  us say a few more words in this direction. Recall the simulations in Section~\ref{sNum}. Much as there, let us assume for the moment that we are working with a band matrix and $W$  is the band width, that is, $a_{ij}=0$ whenever $|i-j|>W$. One can deduce a rather trivial lower bound for the Agmon distance associated to $v$ as in Definition~\ref{distance-def}. If $v_i\geq v_{min}$ for all  $i$ in an interval $I=[i_1, i_{q}]\cup \NN$ of length $q$, then  the Agmon distance between two components of the complement of $I$ 
$$\rho(i_1-1, i_q+1)\geq \Bigl\lfloor\frac{i_q-i_1+1-W}{W}\Bigr\rfloor  \ln  \left(1+\sqrt{\frac{v_{min}}{\max_{i,j\in [N]} a_{ij}}}\right).$$
Here, the lower bracket as usual stands for the floor function. The above inequality follows directly from the observation that the number of non-trivial components such that $v_{i_\ell}\neq 0$ and $v_{i_{\ell+1}}\neq 0$ and $a_{i_\ell i_{\ell+1}}\neq 0$ of the path  from $i_1-1$ to $i_q+1$ 
is at least  $\Bigl\lfloor\frac{i_q-i_1+1-W}{W}\Bigr\rfloor$. Going back to our definitions, and fixing some $\overline E>E$ and the respective partition of $K_{\overline E}=\cup_\ell K_\ell$ into disjoint components, we denote by $d$ the minimal ``Euclidean" distance between the components, i.e., 
$$d:=\min_\ell \min_{i\in K_\ell, \, j\in K_{\ell+1}}|i-j|. $$
It is in our interest to make this distance (or rather  the corresponding Agmon distance) substantial, so we might combine several disjoint components into one $K_\ell$. With this at hand, we choose $\Omega_\ell$ to be maximal possible neighborhoods of $K_\ell$ which are still disjoint. Since the inner boundary $\partial^- \Omega_\ell$ consists of $i\in  \Omega_\ell$ such that $j\not\in \Omega_\ell$ and $a_{ij}\neq 0$, that is, has ``width" at most $W$, one can see that with the aforementioned choices the ``Euclidean" distance between $K_\ell$ and $\partial^-\Omega_\ell$ 
$$\min_{i\in K_\ell, \, j\in \partial^-\Omega_\ell}\geq \frac d2-W,$$
or, to be more precise,  $\Bigl\lfloor\frac{d+1}{2}\Bigr\rfloor-W.$ By design, the  complement of $K_\ell$ in $\Omega_\ell$  consists of  points such that $v_i>\bar E$, that is, there is $v_{min}>0$ such that $v>v_{min}$ in  $\Omega_\ell\setminus (\partial^-\Omega_\ell\cup K_\ell)$. Hence, at the very  least, for this $v_{min}>0$ we have 
$$\rho(\partial^-\Omega_\ell, K_\ell)\geq \Bigl\lfloor\frac{d/2-2W-1}{W}\Bigr\rfloor  \ln  \left(1+\sqrt{\frac{v_{min}}{\max_{i,j\in [N]} a_{ij}}}\right).$$
Clearly, we could take a smaller subinterval of $\Omega_\ell\setminus (\partial^-\Omega_\ell\cup K_\ell)$ and make  $v_{min}>0$ larger, not to mention that this is a trivial lower estimate which does not  take into account  high values of $v$. In any case, this demonstrates that Hypothesis \ref{separation} is non-vacuous.

Let $\psi^j$ denote the complete system of orthonormal eigenvectors of $A$ on $[N]$ with eigenvalues $\lambda_j$. Let $\Psi_{(a,b)}$ denote the orthogonal projection in $\ell^2([N])$ onto the span of eigenvectors $\psi^j$ with eigenvalue $\lambda_j \in (a,b)$. For a fixed $\ell$ let $\vp^{\ell,j}$ denote a complete orthonormal system of the local eigenvectors of $A$ on $\Omega_\ell$ with eigenvalues $\mu_{\ell,j}$, and let $\Phi_{(a,b)}$ be the orthogonal projection onto the span of the eigenvectors $\vp^{\ell,j}$ with eigenvalue $\mu_{\ell,j} \in (a,b)$, over all $\ell$ and $j$.  

The goal of this section is to prove that under the assumption of Hypothesis \ref{separation}, $A$ can be almost decoupled according to  $\bigcup_\ell \Omega_\ell$, with the coupling exponentially small in $S$.   More precisely, we have the following result, which is an analogue of \cite[Theorem 5.1]{ADFJM-math} in the $M$-matrix setting.

\begin{theorem}[Decoupling theorem]\label{tDiag} Assume Hypothesis \ref{separation}. Fix $\delta>0$ and let $\vp$ be one of the local eigenvectors $\vp^{\ell,j}$ with eigenvalue $\mu = \mu_{\ell,j}$ and $\mu \le \overline E - \delta$. Then 
\begin{equation}\label{diag1}
\| \vp - \Psi_{(\mu-\delta, \mu + \delta)}\vp \|^2 \le  \frac{\Conn^2}{\delta^3}\, \max_{i,j \in [N]} |a_{i,j}|^{3}\, e^{-\frac{2S}{\sqrt \Conn}}\|\vp\|^2\,.
\end{equation}
Conversely, if $\psi$ is one of the global eigenvectors $\psi^{j}$ with eigenvalue $\lambda  = \lambda_j\le \overline E - \delta$, then 
\begin{equation}\label{diag2}
\| \psi - \Phi_{(\lambda-\delta, \lambda+ \delta)}\psi \|^2
 \le  \frac{\Conn^2}{\delta^3}\, \max_{i,j \in [N]} |a_{i,j}|^{3}\, e^{-\frac{2S}{\sqrt \Conn}}\|\psi\|^2\,.
\end{equation}
\end{theorem}

\bp We mimic the arguments from Ref.~\citenum{ADFJM-math}. Let us consider the residual vector
$$ r \coloneqq A \vp - \mu \vp = (A - A|_{\Omega_{\ell}}) \vp.$$
Note that the expression $(A - A|_{\Omega_{\ell}}) \vp$ only depends on the values of $\vp$ in the boundary region $\partial \Omega_\ell$.  From Schur's test one thus has
$$\|r\| \leq \Conn^{1/2} \max_{i,j \in [N]} |a_{i,j}|\, \|\vp\|_{\ell^2(\partial\Omega_\ell)}.
$$
We apply Theorem~\ref{t1.1.1} with  $E=\mu$ and $D=\Omega_\ell^c$, so that $K\setminus D=K\cap \Omega_\ell=K_\ell$ and $\rho(k, K\setminus D)=\rho(k, K_\ell) \geq S$  for all $k\in \partial\Omega_\ell$ by Hypothesis \ref{separation}, which then yields
$$\|\vp\|_{\ell^2(\partial\Omega_\ell)}^2\leq e^{-2\alpha S}  \,  \frac{1}{\overline E-\mu} \,\frac \Conn 2 \max_{i,j \in [N]} |a_{i,j}|\, \|\vp\|^2.$$
Taking $\alpha \coloneqq \sqrt{2/\Conn}$ and recalling that $\overline E-\mu>\delta$, we have 
$$\|r\|^2 \leq  \frac{1}{\delta}\, \frac{\Conn^2}{2} \max_{i,j \in [N]} |a_{i,j}|^{3}\, e^{-\frac{2S}{\sqrt \Conn}}\|\vp\|^2. 
$$
From the spectral theorem one has
$$\|r\|^2 \geq \delta^2 \| \vp - \Psi_{(\mu-\delta, \mu + \delta)}\vp \|^2 $$
and the claim \eqref{diag1} follows.

The proof of \eqref{diag2} is somewhat analogous. Let us define the residual vector
$$\widetilde r \coloneqq \sum_\ell (A|_{\Omega_{\ell}} - \lambda I) \psi|_{\Omega_{\ell}} = \sum_\ell I_{\Omega_{\ell}} [A,I_{\Omega_{\ell}}] \psi$$
where the matrices $I_{\Omega_{\ell}}$ were defined in Definition \ref{local}. The values $(I_{\Omega_{\ell}} [A,I_{\Omega_{\ell}}])_{ik}$ are only non-zero when $a_{ik} \neq 0$ and $i, k \in \partial \Omega_\ell$.  In particular, by Hypothesis \ref{separation}, we have $\rho(k,K) \geq S$, and $\tilde r$ only depends on the values of $\psi$ outside of the $S$-neighborhood of $K$. Applying Theorem~\ref{t1.1.1} with $E=\lambda$ and $D=\emptyset$ and applying Schur's test as before, we conclude that
$$\|\widetilde r\|^2 \leq  \frac{1}{\delta}\, \frac{\Conn^2}{2} \max_{i,j \in [N]} |a_{i,j}|^{3}\, e^{-\frac{2S}{\sqrt \Conn}}\|\vp\|^2. 
$$
On the other hand, from the spectral theorem we have
\begin{align*}
\|\widetilde r\|^2 &= \sum_\ell \| (A|_{\Omega_\ell}-\lambda I) \psi_{\Omega_{\ell}} \|^2 \\
&\geq \delta^2 \| \psi - \Phi_{(\lambda-\delta, \lambda+ \delta)}\psi \|^2
\end{align*}
and the claim \eqref{diag2} follows. \ep

The theorem above assures that $A$ can be essentially decoupled on the union of $\Omega_\ell$'s in the sense that the eigenvectors of $A$ are exponentially close to the span of the eigenvectors of $A|_{\Omega_\ell}$, and vice versa. A direct corollary of this result is that the eigenvalues of $A$ are also exponentially close to the combined spectrum of $A|_{\Omega_\ell}$ over all $\ell$:

\begin{corollary} \label{C4.3} Assume Hypothesis \ref{separation}. Fix some $\delta>0$.
Consider the counting functions
\[
N(\lambda) \coloneqq \# \{ \lambda_j:\lambda_j\le \lambda\};
\quad
N_0(\mu) \coloneqq \# \{ \mu_{\ell,j} :  \mu_{\ell,j}  \le \mu\}.
\]
Assume that $\mu\le \overline E$ and choose a natural number $\bar N$ such that 
\[
 \frac{\Conn^2}{\delta^3}\, \max_{i,j \in [N]} |a_{i,j}|^{3}\, \bar N <e^{\frac{2S}{\sqrt \Conn}}. 
\]
Then 
\begin{equation}\label{counting}
\min(\bar N, N_0(\mu - \delta)) \le N(\mu) \quad \mbox{and}
\quad 
\min(\bar N, N(\mu-\delta)) \le N_0(\mu) .
\end{equation}
\end{corollary}

\begin{proof}  Consider the first $p$ unit eigenvectors $\psi^1, \dots, \psi^p$ of $A$, where
$p \coloneqq \min(\bar N, N(\mu-\delta))$. By definition of the counting function, the eigenvalues $\lambda_1,\dots,\lambda_p$ of these eigenvalues are less than $\mu-\delta$. Applying the second conclusion of Theorem \ref{tDiag}, we conclude that
$$\| \psi^k- \Phi_{(0,\mu)}\psi^k \|^2
 \le  \frac{\Conn^2}{\delta^3}\, \max_{i,j \in [N]} |a_{i,j}|^{3}\, e^{-\frac{2S}{\sqrt \Conn}}$$
for $k=1,\dots,p$,
Hence, for any nonzero linear combination $\psi = \sum_{j=1}^p \alpha_j\psi_j $, we have
\begin{align*}
\| \psi - \Phi_{(0,\mu)}\psi\| 
&\leq \sum_j |\alpha_j|  \| \psi^j - \Phi_{(0,\mu)} \psi^j\| \\
&\leq \Big( \frac{\Conn^2}{\delta^3}\, \max_{i,j \in [N]} |a_{i,j}|^{3}\, e^{-\frac{2S}{\sqrt \Conn}} \Big)^{1/2} \|\psi\| p^{1/2}\\
&\leq \Big( \frac{\Conn^2}{\delta^3}\, \max_{i,j \in [N]} |a_{i,j}|^{3}\, e^{-\frac{2S}{\sqrt \Conn}} \bar N \Big)^{1/2}  \|\psi\| \\
&< \|\psi \|.
\end{align*}
It follows that the restriction of $\Phi_{(0,\mu)}$ to the span of the $\psi^j$, $j=1,\, \dots, \,  p$, is injective, and hence the rank $N_0(\mu)$ of the matrix $\Phi_{(0,\mu)}$ is at least $p$. In other words, $N_0(\mu) \ge p$. This establishes the latter inequality in \eqref{counting}; the former one is established similarly.
\end{proof}

% If in two-column mode, this environment will change to single-column format so that long equations can be displayed. 
%\begin{widetext}
%$$\mbox{put long equation here}$$
%\end{widetext}

% Figures should be put into the text as floats. 
% Use the graphics or graphicx packages (distributed with LaTeX2e).
% Put the label that you will use with \ref{} command in the braces of the \label{} command.
%
% \begin{figure}
% \end{figure}

% \begin{table}
% \end{tabular}
% \end{table}

\vspace{5mm}
\noindent{\bf Acknowledgments}
\vspace{5mm}

We are thankful to Guy David for many discussions in the beginning of this project and to Sasha Sodin for providing the literature and context with regard to the edge state localization.

M. Filoche is supported by the Simons foundation grant 601944. S. Mayboroda was partly supported by the NSF RAISE-TAQS grant DMS-1839077 and the Simons foundation grant 563916. T. Tao is supported by NSF grant DMS-1764034 and by a Simons Investigator Award.

\vspace{5mm}
\noindent {\bf Data availability}
\vspace{5mm}

The data that support the findings of this study are available from the corresponding author upon reasonable request.


\bibliographystyle{unsrt}
\bibliography{Filoche-Mayboroda-Tao}

\end{document}